\documentclass{amsart}

\usepackage[T1]{fontenc}
\usepackage[latin9]{inputenc}
\usepackage{color}
\pdfoutput=1
\usepackage{amsthm} 
\usepackage{amssymb}
\usepackage{graphicx}

\usepackage[english]{babel}

\usepackage[authoryear,round]{natbib}

\newtheorem{theorem}{Theorem}

\newtheorem{corollary}{Corollary}
\newtheorem{definition}{Definition}
\newtheorem{example}{Example}

\newtheorem{proposition}{Proposition}
\newtheorem{remark}{Remark}

\providecommand{\tabularnewline}{\\}

\makeatother

\numberwithin{equation}{section}
\numberwithin{figure}{section}

\begin{document}
	
	\title{Evaluating probabilities without model risk}
	
	\author{Joan del Castillo}
	\author{Pedro Puig}
	
	\email[Corresponding author, Joan]{joan.delcastillo@uab.cat}
	\email[Pedro: ]{pere.puig@uab.cat}
	\address[Joan and Pedro: ]{Departament de Matemàtiques, Universitat Autònoma de Barcelona, Cerdanyola del Vallès, Spain}
	\address[Alternate affiliation Pedro Puig: ]{Centre de Recerca Matem\`atica, Cerdanyola del Vallès, Spain.\\ }



	\begin{abstract}
		This article presents methods for estimating extreme probabilities, beyond the range of the observations. These methods are model-free and applicable to almost any sample size. They are grounded in order statistics theory and have a wide range of applications, as they simply require the assumption of a finite expectation. Even in cases when a particular risk model exists, the new methods provide clarity, security and simplicity. The methodology is applicable to the behavior of financial markets, and the results may be compared to those provided by extreme value theory.
	\end{abstract}

		\maketitle

\section{Introduction}

The main objective of this work is to determine how to evaluate the
probabilities of a random phenomenon in the right tail of its distribution,
without making assumptions about the model that produced it. In
this direction, an entirely new inequality is introduced that is extremely
well-founded and has a wide range of applications for extrapolating
into the tails.

Non-parametric methods that rely on the distribution of order statistics,
often provide results that are completely free of model risk. However, they
do not provide information beyond an observed sample. This theory
could be said to begin with \cite{Wilks} and is explained in \cite{ecastillo88,David}
and \cite{Gibbons} among many others. From a probabilistic point
of view, Markov's inequality describes the decrease of the complementary
cumulative distribution function in terms of its expected value. A
recent improvement to the inequality starts with the concept of partial
expectation by \cite{Cohen} and is developed by \cite{Ogasawara}.
This provides inequalities that require no additional assumptions
and are much more powerful than the original for high thresholds.

Our job has been to apply these theories to
reality, by identifying and analyzing the properties of an estimator
for the partial expectation. The inequalities (\ref{eq:iB1}) and
(\ref{eq:iBa}) represent the article's primary findings. Its foundations,
which are summarized in Corollary \ref{cor:Cor} and Proposition \ref{prop:Exten},
holds true regardless of model and for virtually any sample size,
as seen in Section \ref{sec:S2}.

The third key element in our work is the extreme value theory
explained in \cite{Coles}, \cite{McNeil}, \cite{Beirlant} and \cite{davison15},
among numerous other authors. Specifically, this theory proves that,
under broadly applicable regularity requirements, the residual distribution
of any heavy-tailed distribution asymptotically resembles a Pareto
distribution over a high enough threshold, see Section \ref{subsec:Asymptotic}.

Section \ref{sec:Num} uses simulations to analyze the differences between
the inequalities (\ref{eq:iB1}) and (\ref{eq:iBa}). Table \ref{tab:T1}
shows that while the second inequality is more stable, it frequently
overestimates the tail function for $a=3$ and $a=5$, whereas for
$a=1$ the bound is closer to the true probability. Table \ref{tab:T2}
shows that the higher the quantile, the safer the use of inequality
(\ref{eq:iB1}). It may  be observed that the results change very little
between the sample sizes (from $10$ to $1000$), but vary in relation
to the number of finite distribution moments.

In Section \ref{sec:DJI}, the results are applied to the
behavior of the Dow Jones Industrial Average equities index, during a period that includes the 2020 Covid stock market crisis. The probability of extreme losses are calculated with the inequality (\ref{eq:iB1}) and with three
estimates of the location Pareto distribution made with the theory
of extreme values, see Figure \ref{fig:LP_fit}. We compared the results of the two methodologies and highlighted the manner in which they are interrelated and complementary, providing a more comprehensive understanding of the process of extrapolation from observed levels to unobserved levels. The paper concludes with discussion.

\section{\protect\label{sec:S2}Several developments in fundamental probability and statistics}

Using the empirical cumulative distribution function, the probability
of finding values larger than the maximum observed in a sample is
zero. It is also illogical to think that this is true, see \cite{Carvalho}.
However, as we shall see, there is a fix for this, which essentially involves presuming that the data analysis can make use of the law of large numbers.

In this section we assume that $X$ is a nonnegative random variable
with cumulative distribution function $F\left(x\right)$. Henceforth
it will be called tail function to $1-F\left(x\right)=Pr\left\{ X>x\right\} $. 

\subsection{Improved Markov's inequality}

In addition to the above assumptions, we now also add that  $X$ has a finite expectation. According to the traditional Markov's inequality the probability
that $X$ is at least $\nu>0$ is at most the expectation of $X$
divided by $\nu$. That is
\begin{equation}
Pr\left\{ X>\nu\right\} \leq E\left(X\right)/\nu.\label{eq:traditional}
\end{equation}
If the random variable $X$ has $k>1$ finite moments, applying Markov's
inequality to $X^{k}$ yields
\begin{equation}
Pr\left\{ X>\nu\right\} =Pr\left\{ X^{k}>\nu^{k}\right\} \leq E\left(X^{k}\right)/\nu^{k},\label{eq:momen}
\end{equation}
which ensures that the tail function decreases as $\nu^{-k}$, or
faster.

Recent improvements to Markov's inequality start with the concept
of partial expectation by \cite{Cohen} and are developed by \cite{Ogasawara}.
One of its results follows from the next definition.
\begin{definition}
The partial expectation of $X$ above $\nu>0$ is defined as $E_{\nu}\left(X\right)=E\left(X1_{\left\{ X>\nu\right\} }\right)$,
where $1_{\left\{ X>\nu\right\} }$ is the indicator of the event
$\left\{ X>\nu\right\} $. 

For clarity, if it is assumed that $X$ has a probability density
function $f\left(x\right)$, so its partial expectation is expressed
as 
\[
E_{\nu}\left(X\right)=E\left(X1_{\left\{ X>\nu\right\} }\right)=\int_{\nu}^{\infty}x\,f(x)\,dx.
\]
This concept is employed, for instance, in the Lorenz curve, \cite{Arnold},
for more details see \cite{Cohen}. 
\end{definition}

\begin{theorem}
\label{markov}If $X$ is a nonnegative random variable and $\nu>0$,
then the probability of $X$ being at least $\nu$ is at most the
partial expectation of X over $\nu$ divided by $\nu$. 
\begin{equation}
\Pr\left\{ X>\nu\right\} \leq E_{\nu}\left(X\right)/\nu.\label{eq:Markov}
\end{equation}
\end{theorem}

\begin{proof}
From the inequalities $\nu1_{\left\{ X>\nu\right\} }\leq X1_{\left\{ X>\nu\right\} }$
and taking the expectation of both sides the result follows.
\end{proof}
Moreover, if $\nu>\nu_{0}\geq0$, taking the expectation on inequality
$X1_{\left\{ X>\nu\right\} }\leq X1_{\left\{ X>\nu_{0}\right\} }$
produces

\begin{equation}
\Pr\left\{ X>\nu\right\} \leq E_{\nu}\left(X\right)/\nu\leq E_{\nu_{0}}\left(X\right)/\nu,\label{eq:creixent}
\end{equation}
In particular for $\nu_{0}=0$, (\ref{eq:creixent}) shows that the
inequality (\ref{eq:Markov}) improves the traditional Markov\textquoteright s
inequality (\ref{eq:traditional}). The inequality (\ref{eq:Markov}),
hereinafter called the improved Markov's inequality, requires no additional
assumptions and is hundreds or thousands of times more powerful than
the original for high thresholds (see Examples \ref{exa:expon} and
\ref{exa:HNor}).

The following Proposition and Example \ref{exa:ExPareto} show that
the improved Markov bound and the tail function decrease at the same
order, which is linked to the number of finite distribution moments,
as opposed to the traditional Markov upper limit, which declines as $1/\nu$.
\begin{proposition}
\label{prop:P7}Let $X$ be a nonnegative random variable with cumulative
distribution function $F\left(x\right)$, probability density function
$f\left(x\right)$ and $\nu>0$. Assuming that $X$ has a finite expectation,  the error between the tail function and the improved
Markov bound is  
\begin{equation}
E_{\nu}\left(X\right)/\nu-\Pr\left\{ X>\nu\right\} =\frac{1}{\nu}\int_{\nu}^{\infty}\left(1-F\left(x\right)\right)\,dx\label{eq:error}
\end{equation}
\end{proposition}

\begin{proof}
Computing the integral, using integration by parts,

\begin{multline*}
E_{\nu}\left(X\right)=\int_{\nu}^{\infty}x\,f\left(x\right)\,dx=\left[x\,\left(F\left(x\right)-1\right)\right]_{\nu}^{\infty}-\int_{\nu}^{\infty}\left(F\left(x\right)-1\right)\,dx\\
=\nu\left(1-F\left(\nu\right)\right)+\int_{\nu}^{\infty}\left(1-F\left(x\right)\right)\,dx,
\end{multline*}
then 
\[
E_{\nu}\left(X\right)/\nu=\left(1-F\left(\nu\right)\right)+\frac{1}{\nu}\int_{\nu}^{\infty}\left(1-F\left(x\right)\right)\,dx,
\]
and the result follows.
\end{proof}
In general if $X$ has $k$ finite moments, from (\ref{eq:momen}),
$\left(1-F\left(x\right)\right)$ decreases as $\nu^{-k}$. and $xf\left(x\right)$
is of the same order. Moreover, in (\ref{eq:error}) the integral
in the error term decreases by one less unit, but the error by the
same order as the tail function and the improved Markov bound. This is significant since the partial expectation already shows what the traditional Markov's inequality does with the higher moments. 
This explains why the improved Markov result of Theorem \ref{markov}
is very precise. The above Proposition is illustrated by an example.

The Pareto type I distribution, \cite{Arnold},
also known in physics as power law distribution,
is the two-parameter family of continuous probability density functions

\begin{equation}
f_{pl}\left(x;\alpha,\mu\right)=\frac{\alpha}{\mu}\left(\frac{\mu}{x}\right)^{\alpha+1},\;\left(x\geq\mu\right),\label{eq:PL}
\end{equation}
where the shape parameter $\alpha>0$ is known as the tail index,
and the parameter $\mu>0$ is the minimum value  of the distribution support $\left(\mu,\infty\right)$. 
\begin{example}
\label{exa:ExPareto}Let $X$ be distributed as (\ref{eq:PL}), with
$\mu=1$ and support $\left(1,\infty\right)$.

The moments of the distribution are given by 
\[
E\left(X^{k}\right)=\alpha\int_{1}^{\infty}\frac{x^{k}}{x^{\alpha+1}}\,dx=\alpha\left[\frac{x^{k-\alpha}}{k-\alpha}\right]_{1}^{\infty}=\begin{cases}
\frac{\alpha}{\alpha-k}<\infty & k<\alpha\\
=\infty & k\geq\alpha
\end{cases},
\]
only for $k<\alpha$, the distribution has finite moments. 

Assuming $\alpha>1$, the straightforward integration shows that the partial
expectation is 
\[
E_{\nu}\left(X\right)=\int_{\nu}^{\infty}\frac{\alpha}{x^{\alpha}}\,dx=\frac{\alpha}{\left(\alpha-1\right)\,}\nu^{1-\alpha}.
\]
Then, the improved Markov bound is $E_{\nu}\left(X\right)/\nu=\alpha\left(\alpha-1\right)^{-1}\,\nu^{-\alpha}$,
which decreases as the tail $1-F\left(\nu\right)=\nu^{-\alpha}$.
The difference, the error term, is $\left(\alpha-1\right)^{-1}\,\nu^{-\alpha}$, 
which decreases at the same order as the two terms. It fits the result
of Proposition \ref{prop:P7} .
\end{example}

\subsection{Fundamental inequality}

The improved Markov's inequality, really very accurate, leads to trying
a practical application.

Assume that the phenomenon under consideration may be represented
by $X$, a nonnegative continuous random variable with a finite expectation. A sample from $X$ is a sequence $X_{1},X_{2},...,X_{n}$
of independent random variables identically distributed as $X$
(i.i.d). The values of the sequence rearranged in an increasing order
are denoted by $X_{1,n}\leq X_{2,n}\leq...\leq X_{n,n}$, then $X_{r,n}$
is called the $r$-th order statistic ($r=1,...,n$).

It is assumed here that a sample of size $n$ has been observed. Their
ordered values are represented by $0\leq x_{1,n}\leq x_{2,n}\leq...\leq x_{n,n}$,
and every new event is still represented by $X$. The challenge is
to find a practical bound of the tail function $Pr\left\{ X>\nu\right\} $
even when $\nu>x_{n,n}$, from the improved Markov bound $E_{\nu}\left(X\right)/\nu$.

Since sample mean is a good estimator for the expected value, if
$\nu\leq x_{n,n}$, the partial expectation $E_{\nu}\left(X\right)$,
as the expectation of $X1_{\left\{ X>\nu\right\} }$, may be estimated
by the average of the sample in which values less than $\nu$ are
nullified. If we let $k$ be the number of observations larger than $\nu$,
$\nu\leq x_{n,n}$, $E_{\nu}\left(X\right)$ may be estimated
by the partial mean, $pM_{k}$,
\begin{equation}
pM_{k}=\frac{1}{n}\sum_{j=0}^{k-1}x_{n-j,n}.\label{eq:rMean}
\end{equation}
If $\nu_{0}\leq x_{n,n}<\nu$, the partial expectation $E_{\nu_{0}}\left(X\right)$
may  be estimated from the sample, but $E_{\nu}\left(X\right)$ may not,
and (\ref{eq:creixent}) sets the upper limit
\begin{equation}
\Pr\left\{ X>\nu\right\} \leq E_{\nu_{0}}\left(X\right)/\nu,\label{eq:upper}
\end{equation}
As we are interested in a small upper bound, to calculate $E_{\nu_{0}}\left(X\right)$,
we will choose as large a $\nu_{0}$ as possible provided that it may be estimated. For $x_{n-1,n}\leq\nu_{0}<x_{n,n}$, the partial mean is
constant $pM_{k_{0}}=x_{n,n}/n$, and this is the smallest estimate
of $E_{\nu_{0}}\left(X\right)$ that may be made from the sample.
This leads to the following inequality for $\nu\geq x_{n,n}$

\begin{equation}
Pr\{X>\nu\}\leq\frac{x_{n,n}}{n\,\nu}.\label{eq:iB1}
\end{equation}

If $a>1$, a less demanding inequality is 
\begin{equation}
Pr\{X>\nu\}\leq\frac{ax_{n,n}}{n\,\nu}.\label{eq:iBa}
\end{equation}

An elementary empirical approximation of inequality (\ref{eq:Markov})
is given by inequality (\ref{eq:iB1}). At this point it is once again necessary to consider random $x_{n,n}$ in order to study its features. This
will be done below. The empirical upper bound of the tail
function in (\ref{eq:iB1}) is denoted by
\begin{equation}
eB\left(\nu\right)\equiv eB\left(\nu,n,X_{n,n}\right)=\frac{X_{n,n}}{n\,\nu}.\label{eq:bound}
\end{equation}

\subsection{On the distribution of the maximum}

Assume now that $X$ is an absolutely continuous nonnegative random
variable with probability density function $f\left(x\right)=F'\left(x\right)$
and quantile function $Q\left(p\right)=F^{-1}\left(p\right)$ $(0<p<1)$.
Let $X_{1:n}\leq X_{2,n}\leq...\leq X_{n:n}$ be a sample from $X$ rearranged
in increasing order (unobserved).

The variable $X$ may be converted by probability transform $U=F\left(X\right)$
into a standard uniform distribution $U(0,1)$ with cumulative distribution
function $F_{U}\left(x\right)=x$. The classical non-parametric theory
extends the result to order statistics, $F\left(X_{r,n}\right)=U_{r,n}$,
determining that these have beta distribution, $U_{r,n}\sim Beta\left(r,n-r+1\right)$,
see \cite{ecastillo88}. We may then calculate the moments from those
of the beta distribution
\begin{equation}
E\left[F\left(X_{r,n}\right)\right]=E\left[U_{r,n}\right]=\frac{r}{n+1},\label{eq:Emax}
\end{equation}
In particular this outcome applies to the distribution of the maximum
of the sample, $X_{n,n}$. In this case the result is well known. We include it to make the article self-contained, by using it in the next Corollary.
\begin{proposition}
\label{prop:distribucio}The cumulative distribution function of the
absolutely continuous random variable $X$ evaluated over the maximum
of a sample $X_{n,n}$ does not depend on the distribution itself.
More specifically, for $0<x<1$,
\[
Pr\left\{ F\left(X_{n,n}\right)\leq x\right\} =x^{n}.
\]
\end{proposition}

\begin{proof}
Since $F_{U}\left(x\right)=x$, and $U_{n,n}$ is the maximum
of $n$ standard uniform random variables, it follows that  $F_{U_{n,n}}\left(x\right)=F_{U}\left(x\right)^{n}=x^{n}$.
\end{proof}
If we let $X_{n,n}$ be the maximum of a $n$-sample from $X$, and let $X=X_{n+1}$ be
a new event, from (\ref{eq:Emax}) the expected value of the random
variable $Pr\left\{ X>X_{n,n}\right\} =1-F\left(X_{n,n}\right)$ is
$1/(n+1)$. Classical non-parametric theory uses this value as a first
approximation to the variable, see \cite{ecastillo88}.

\begin{equation}
Pr\left\{ X>X_{n,n}\right\} \approx\frac{1}{n+1}.\label{eq:NPin}
\end{equation}

\begin{remark}
The inequality (\ref{eq:iB1}) for $\nu=x_{n,n}$ coincides with 
the classical result (\ref{eq:NPin}), except for factor $(n+1)/n$.
This coincidence, which occurred after two very different explanations,
supports the new inequality which may be applied for all values of
$\nu$, unlike the previous one.
\end{remark}

The next result, which does not require a finite expectation, goes
further in order to justify the inequality (\ref{eq:iBa}).
\begin{corollary}
\label{cor:Cor}The probability that the tail function of an absolutely
continuous nonnegative random variable exceeds the maximum of a sample
in $a/n<1$ is
\[
Pr\left\{ Pr\{X>X_{n,n}\}>a/n\right\} =Pr\left\{ 1-F\left(X_{n,n}\right)>a/n\right\} =\left(1-a/n\right)^{n}\approx e^{-a}
\]
in particular,\textup{ }for $a=5$\textup{
\[
Pr\left\{ 1-F\left(X_{n,n}\right)<5/n\right\} =1-(1-5/n)^{n}\approx0.993,
\]
} and for $a=1$

\textup{
\[
Pr\left\{ 1-F\left(X_{n,n}\right)<1/n\right\} =1-\left(1-1/n\right)^{n}\approx(1-1/e)=0.632,
\]
}
\end{corollary}

\begin{proof}
Since $Pr\left\{ 1-F\left(X_{n,n}\right)>a/n\right\} =Pr\left\{ F\left(X_{n,n}\right)<1-a/n\right\} $,
from Proposition \ref{prop:distribucio}
\[
Pr\left\{ F\left(X_{n,n}\right)<1-a/n\right\} =\left(1-a/n\right)^{n}\approx e^{-a}.
\]
\end{proof}
\begin{remark}
These results are significant because (\ref{eq:iBa}) holds true regardless
of the model and for virtually any sample size (say $n>2a$), since $1-\left(1-a/n\right)^{n}$is
a decreasing sequence up to the limit $1-e^{-a}$. For $n=10$, $1-(1-5/n)^{n}=0.999$, a number that decreases as the sample size increases to $1-e^{-5}=0.99326$,
really close. Note that for $a=1$, the above result ensures that the
inequality (\ref{eq:iB1}) holds for $\nu=x_{n,n}$ in $63.2\%$ of
the cases. For $a=3$, (\ref{eq:iBa}) is guaranteed with with a confidence level of  $95.0\%$ and, for $a=5$, with a confidence level of  $99.3\%$. 
\end{remark}

Generally speaking, greater values also satisfy the inequality (\ref{eq:iBa})
if it is satisfied at a certain point. Let us see a specific situation.
\begin{proposition}
\label{prop:Exten}Let $X$ be a random variable for which inequality
(\ref{eq:iBa}) holds at \textup{$x_{0}$,}

\[
Pr\{X>x_{0}\}\leq\frac{ax_{n,n}}{n\,x_{0}}.
\]
If its tail function is such that $x\left(1-F\left(x\right)\right)$is
a decreasing function for \textup{$x>x_{0}$, }then the inequality
holds for all \textup{$x>x_{0}$.}
\end{proposition}

\begin{proof}
Let $c\left(x\right)=eB\left(x\right)/\left(1-F\left(x\right)\right)=ax_{n,n}/(n x\left(1-F\left(x\right)\right))$,
which is an increasing function for $x>x_{0}$, from the assumption.
Since (\ref{eq:iBa}) holds at $x_{0}$, $c\left(x_{0}\right)\geq1$,
then $c\left(x\right)>1$ for $x>x_{0}$. Hence, $1-F\left(x\right)<eB\left(x\right)$ for all $x>x_{0}$.
\end{proof}
 
Pareto I distribution (\ref{eq:PL}) with $\alpha>1$ is covered by
the above Proposition. We simply examine the standard distribution
of Example \ref{exa:ExPareto}, with tail function $1-F(x)=x^{-\alpha}$,
since the inequality is unaffected by scale changes. If $\alpha>1$,
then $x\left(1-F(x)\right)=x^{-(\alpha-1)}$ is a decreasing function
for all $x$ in the support $\left(1,\infty\right)$. Note that in this case  $X$ has a finite expectation if and only if $x\left(1-F\left(x\right)\right)$ is a decreasing function.

The Proposition \ref{prop:Exten} applies to exponential distribution
for $x$ larger than the expected value, since for the standard case $x\,exp(-x)$ decreases for
$x>1$, see Example \ref{exa:expon}. The Proposition also applies
to a half-normal distribution, since for the standard case with unit
expectation, $x\left(1-F\left(x\right)\right)$is a decreasing function
for $x>0.9423$, see Example \ref{exa:HNor}.

Even so, we are interested in showing that in practice we can also
apply the inequality for $a=1$. The reason is that $a\,eB\left(x\right)$
frequently overprotects by an unnecessary factor $a$, whereas $eB\left(x\right)$
is a bound adjusted to the true probability, as we will discuss in Section
\ref{sec:Num}.

The relationship between the maximum of a $n$-sample and the quantile
$(n-1)/n$ is examined in order to determine when inequality (\ref{eq:iB1})
is satisfied.
\begin{proposition}
\label{prop:P10}For any $n$-sample of the variable $X$ and for
\textup{$q_{1}=Q\left(1-1/n\right)$} , the next three claims are
equivalent 
\end{proposition}

\begin{enumerate}
\item The maximum of the sample exceeds $q_{1}$, that is $X_{n,n}\geq q_{1}$.
\item The inequality (\ref{eq:iB1}) is satisfied in $X_{n,n}$, $1-F\left(X_{n,n}\right)\leq1/n$.
\item The inequality (\ref{eq:iB1}) is satisfied in $q_{1}$, $\left(1-F\left(q_{1}\right)\right)\leq eB\left(q_{1}\right)$. 
\end{enumerate}
\begin{proof}
The results are deduced from the following equivalences:
\[
X_{n,n}\geq q_{1}\Leftrightarrow F\left(X_{n,n}\right)\geq F\left(q_{1}\right)=1-1/n\Leftrightarrow1-F\left(X_{n,n}\right)\leq1/n
\]

\[
eB\left(X_{n,n}\right)-\left(1-F\left(X_{n,n}\right)\right)=1/n-\left(1-F\left(X_{n,n}\right)\right)\geq0
\]

\[
X_{n,n}\geq q_{1}\Leftrightarrow\frac{X_{n,n}}{n\,q_{1}}-\frac{1}{n}\geq0\Leftrightarrow eB\left(q_{1}\right)-\left(1-F\left(q_{1}\right)\right)\geq0
\]
\end{proof}
Consequently, regardless of the model and for virtually any sample size,
the following three equivalent claims are satisfied in $63.2\%$ of
the cases
\[
X_{n,n}\geq q_{1}\Leftrightarrow1-F\left(q_{1}\right)\leq\frac{X_{n,n}}{n\,q_{1}}\Leftrightarrow1-F\left(X_{n,n}\right)\leq1/n,
\]
and the inequality (\ref{eq:iB1}) is satisfied at $X_{n,n}$ and
at $q_{1}$.

In practice, if the sample being examined may  be assumed to have a
very high maximum, it will likely surpass the median of maximums
and the inequality will be satisfied. Only when randomness yields
a sample with a maximum within $37\%$ of the smallest values (below
the quantile $1/e$), well below the median of the maximums, will
it cease to work, see also Section \ref{sec:DJI}.

\section{\protect\label{subsec:Asymptotic}Asymptotic methods}

Inequalities (\ref{eq:iB1}) and (\ref{eq:iBa}) based on Corollary
\ref{cor:Cor} and Proposition \ref{prop:Exten} may be applied in
almost any circumstance. They provide relevant information and, are very
reliable and easy to use.

By closely examining the above passage, we can further state that
while the Corollary is independent of sample size and variable distribution,
it must be assumed that the sample has been selected appropriately (i.i.d. observations).
Proposition \ref{prop:Exten} needs a condition similar to finite expectation (equivalent in some models), with no need for finite variance. This is as to assuming the law of large numbers.

Additionally, the Proposition requires that the sample be large enough
to ensure that the maximum exceeds the value for which $x\left(1-F\left(x\right)\right)$
is a decreasing function. It has been observed that there are no constraints
on the Pareto I distribution. Simulation ($10^{5}$
samples) has shown us that the exponential distribution needs a sample size $n\geq10$ to ensure with $99\%$ confidence that the maximum exceeds the threshold
from which the function decreases. For a half-normal distribution
a sample size of $n\geq8$ is needed to ensure the same  with $99\%$ confidence.

Applications often require the tough challenge of probability estimate
beyond the observed sample. This suggests extrapolating observed levels
in the direction of unobserved levels. In this regard there are no
serious competing  models to those provided by extreme value theory.

Extreme value theory provides a large class of models that allow extrapolation
into the tails of a distribution, see \cite{davison15}, \cite{Beirlant}
and \cite{Coles}. The theory is divided into two branches. The first examines
the behavior of block maxima and characterizes their normalized limits
as members of the generalized extreme value distribution, which includes
the Gumbell, Weibull, and Fréchet distributions. The second, originated
in the 1970s, makes it possible to deduce from the previous characterization
the asymptotic behavior of the residual distribution over a high threshold
as members of the generalized Pareto distribution. Both branches use
asymptotic arguments and it may
be necessary to extrapolate into unseen levels. Hence, there are still
a number of issues in the applications, even in the univariate case,
see \cite{Gomes}.

We will now consider random variables that have an unbounded right
tail and some finite moments, or all finite moments as exponential
variables. If a distribution has only a finite number of finite moments,
it is said to be heavy-tailed. Extreme value theory guarantees that
asymptotically the exceedances above a large threshold $\nu$ (the
tail) of any heavy-tailed distribution follow a classical Pareto distribution,
under widely applicable regularity conditions, see \cite{Coles} and
\cite{Beirlant}. 

This methodology partially avoids model risk, but only asymptotically,
because convergence is slow at times and it faces a threshold selection
problem. The new approach, with inequalities (\ref{eq:iB1}) and
(\ref{eq:iBa}), is a complement to the extreme value theory and reinforces or makes it possible to discuss the results. 

\subsection{\protect\label{subsec:Pareto}Location Pareto distribution}

The Pareto type II distribution, \cite{Arnold}, 
has probability density function

\begin{equation}
f_{lpd}\left(x;\alpha,\mu,\delta\right)=\frac{\alpha}{\mu+\delta}\left(\frac{\mu+\delta}{x+\delta}\right)^{\alpha+1},\;\left(x\geq\mu\right),\label{eq:LPD}
\end{equation}
where the shape parameter $\alpha>0$ is the tail index, the parameter
$\mu$ is the minimum value of
the distribution support $\left(\mu,\infty\right)$ and $\delta>-\mu$
is the change in the unit of measurement's origin. For non-negative
values of $\mu\geq0$, we shall refer to the model (\ref{eq:LPD})
location Pareto distribution, which  will be represented by $LPD\left(\alpha,\mu,\delta\right)$ (negative
values of $\mu$ are also allowed in Pareto type II). This model includes
the power law (\ref{eq:PL}) for $\delta=0$. Additionally, for a
higher threshold $\mu$ the two models perform similarly, to distinguish between these two models see \cite{Castillo23}.

The LPD also includes the classical Pareto distribution
for $\mu=0$. Its probability density function often adopts an alternative
parameterization.
\begin{equation}
f_{pd}\left(x;\alpha,\delta\right)=\frac{\alpha}{\delta}\left(\frac{\delta}{x+\delta}\right)^{\alpha+1}=\frac{1}{\beta}\left(1+\frac{x\xi}{\beta}\right)^{-\left(1+\xi\right)/\xi},\;\left(x\geq0\right),\label{eq:PD}
\end{equation}
where $\xi=1/\alpha$ is the extreme value index and $\beta=\delta/\alpha$
is scale.

The Pareto distribution (\ref{eq:PD}), with the parameterization
$\left(\xi,\beta\right)$, generalizes to upper bounded distributions
with the previous formula (\ref{eq:PD}) for $\xi<0$, and to the
exponential distribution as a limiting case when $\xi$ tends to $0$.
Then, the joint model is called the generalized Pareto distribution.
The case of interest for us is $\xi>0$ (heavy tails), Since the larger
integer less than $\alpha=1/\xi$ indicates the number of finite moments
in the model, this parameter will be used often.

Fitting models (\ref{eq:LPD}) or (\ref{eq:PL}) to a sample entails
the problem of finding a reasonable (unambiguous and automatic) threshold
within which the theory may be applied. This is a threshold selection problem
rather than an estimation problem, which is often performed visually with several graphs. The techniques proposed by \cite{Clauset}
present an automatic method for choosing the threshold $\mu$ that provides
the best asymptotic approximation to the power law distribution. Alternatively,
\cite{Castillo23} and \cite{Castillo19} provide an automatic threshold
selection algorithm (\emph{thrselect} function in the R package \emph{ercv})
for location Pareto distribution (\ref{eq:LPD}), using a residual
variation coefficient technique.

\section{\protect\label{sec:Num}Numerical evidence}

Asymptotic arguments and mathematical proofs underpin the conclusions
in Section \ref{sec:S2}. However, when practical applications
are considered, doubts always arise. The validity of the inequalities
(\ref{eq:iB1}) and (\ref{eq:iBa}) requires the existence of a finite
moment in the model that produced the data. How do we know if this
is accurate?

According to  experts in the physics of complex systems most identified
power laws in nature have exponents such that the mean is well-defined
but the variance is not $\left(1<\alpha<2\right)$, see \cite{Clauset}.
Several authors claim that financial returns have finite skewness
but infinite kurtosis $\left(3<\alpha<4\right)$. In any case, what
is done, under various formats, is to apply the extreme value theory,
as we will do.

If estimating the LPD distribution gives a value $\alpha>1$,
from the maximum of the studied sample the behavior of the data must
be close to the Pareto distribution and, if $\alpha$ is large enough,
near the exponential distribution. Therefore, we will use these distributions
to see certain aspects of the previous inequalities. In addition, the
half-normal distribution is included, which takes into account lighter
tails than the exponential. 

We have simulated these distributions to observe how the number of
finite moments affects inequalities (\ref{eq:iB1}) and (\ref{eq:iBa}),
and how they behave beyond $X_{n,n}$. Tables \ref{tab:T1} and
\ref{tab:T2} show the outcomes.
\begin{example}
\label{exa:expon}The cumulative distribution function and the probability
density function of the standard exponential distribution, $X$, are
given by 
\[
F\left(x\right)=1-e^{-x},\;f\left(x\right)=e^{-x},\;x\geq0.
\]
The partial expectation is easy to calculate using integration
by parts 
\[
E_{\nu}\left(X\right)=\int_{\nu}^{\infty}x\,e^{-x}\,dx=e^{-\nu}\left(1+\nu\right),
\]
and if $\nu=0$, the expectation is $E\left(X\right)=1$.
\end{example}

The $0.99$ and $0.999$ quantiles for $X$ are $4.605$ and $6.908$.
The three terms of the inequalities (\ref{eq:creixent}) for $\nu=6.908$
and $\nu_{0}=0$ are: the tail function $1.00\,10^{-3}$, the improved\emph{
}Markov bound $1.14\,10^{-3}$ and the traditional Markov bound $0.145$.
Simply by looking at the orders of magnitude, the difference is obvious, as it is around one thousand times better.
\begin{example}
\label{exa:HNor}If $X$ is normally distributed with zero mean, $N\left(0,\sigma^{2}\right)$,
then its absolute value $\left|X\right|$ follows a \emph{half-normal}
distribution, supported on the interval $[0,\infty)$, with the cumulative
distribution function
\[
F_{H}\left(x;\sigma\right)=2\Phi\left(x/\sigma\right)-1,\;x\geq0.
\]
where $\Phi\left(x\right)$ is the cumulative distribution function
of the standard normal distribution. The probability density function
is
\[
f_{H}\left(x;\sigma\right)=\frac{2}{\sqrt{2\pi}\sigma}exp\left(-\frac{x^{2}}{2\sigma^{2}}\right),\;x\geq0.
\]
The partial expectation of the half-normal distribution is given by
\[
E_{\nu}\left(\left|X\right|\right)=\frac{\sqrt{2}}{\sqrt{\pi}\sigma}\int_{\nu}^{\infty}x\,exp\left(-\frac{x^{2}}{2\sigma^{2}}\right)\,dx=\sigma\sqrt{2/\pi}exp\left(-\frac{\nu^{2}}{2\sigma^{2}}\right),
\]
then in particular $E\left(\left|X\right|\right)=\sigma\sqrt{2/\pi}$.

Then, a half-normal distribution $Z$ where $\sigma=\sqrt{\pi/2}$
has unit expectation, $E\left(Z\right)=1$, and its partial expectation
is 
\[
E_{\nu}\left(Z\right)=exp\left(-\nu^{2}/\pi\right),
\]
\end{example}

The $0.99$ and $0.999$ quantiles for $Z$ are $3.228$ and $4.124$.
The three terms of the inequalities (\ref{eq:creixent}) for $Z$
and $\nu=4.124$ and $\nu_{0}=0$ are: the tail function $1.00\,10^{-3}$,
the improved\emph{ }Markov bound $1.08\,10^{-3}$ and the traditional
Markov bound $0.242$. The new inequality is thousands times better.

The lower bound of the improved Markov bound decreases at the same rate
as the tail function, as demonstrated by Proposition \ref{prop:P7}
and Example \ref{exa:ExPareto} in Section \ref{sec:S2}. As opposed
to the traditional Markov upper limit, which declines as $1/\nu$.

\begin{table}
\caption{\protect\label{tab:T1}Tail function slightly overfitted by the empirical
upper bound (\ref{eq:bound}), $eB\left(\nu\right)$. }
\resizebox{\columnwidth}{!}{%
\begin{tabular}{|c|c|c|c|c|c|c|c|c|c|c|}
\hline 
\multicolumn{11}{|c|}{Simulation of half-normal distribution (expectation $1$, samples
$10^{5}$)}\tabularnewline
\hline 
\hline 
samp. size & $(n-1)/n$ & \multicolumn{5}{c|}{eB empirical quantiles ($n$ times)} & sample & \multicolumn{3}{c|}{probability exceeding tail}\tabularnewline
\hline 
$n$ & $F$-quantile & min & 0.01 & median & 0.99 & max & mean & $a=1$ & $a=3$ & $a=5$\tabularnewline
\hline 
1000 & 4.125 & 0.782 & 0.863 & 1.032 & 1.342 & 1.839 & 1.045 & 0.633 & 0.952 & 0.994\tabularnewline
\hline 
100 & 3.229 & 0.623 & 0.780 & 1.050 & 1.506 & 2.113 & 1.067 & 0.637 & 0.954 & 0.995\tabularnewline
\hline 
10 & 2.062 & 0.237 & 0.550 & 1.114 & 1.999 & 3.236 & 1.144 & 0.652 & 0.973 & 1.000\tabularnewline
\hline 
\multicolumn{11}{|c|}{Simulation of exponential distribution (expectation $1$, samples
$10^{5}$)}\tabularnewline
\hline 
samp. size & $(n-1)/n$ & \multicolumn{5}{c|}{eB empirical quantiles ($n$ times)} & sample & \multicolumn{3}{c|}{probability exceeding tail}\tabularnewline
\hline 
$n$ & $F$-quantile & min & 0.01 & median & 0.99 & max & mean & $a=1$ & $a=3$ & $a=5$\tabularnewline
\hline 
1000 & 6.908 & 0.622 & 0.779 & 1.054 & 1.666 & 2.689 & 1.085 & 0.636 & 0.952 & 0.994\tabularnewline
\hline 
100 & 4.606 & 0.397 & 0.676 & 1.080 & 2.001 & 3.636 & 1.126 & 0.632 & 0.953 & 0.995\tabularnewline
\hline 
10 & 2.303 & 0.186 & 0.434 & 1.176 & 2.998 & 6.193 & 1.274 & 0.653 & 0.972 & 0.999\tabularnewline
\hline 
\multicolumn{11}{|c|}{Simulation of Pareto distribution ($\alpha=6$, samples $10^{5}$)}\tabularnewline
\hline 
samp. size & $(n-1)/n$ & \multicolumn{5}{c|}{eB empirical quantiles ($n$ times)} & sample & \multicolumn{3}{c|}{probability exceeding tail}\tabularnewline
\hline 
$n$ & $F$-quantile & min & 0.01 & median & 0.99 & max & mean & $a=1$ & $a=3$ & $a=5$\tabularnewline
\hline 
1000 & 3.163 & 0.675 & 0.777 & 1.064 & 2.144 & 7.048 & 1.130 & 0.634 & 0.951 & 0.994\tabularnewline
\hline 
100 & 2.155 & 0.668 & 0.778 & 1.064 & 2.171 & 8.389 & 1.130 & 0.634 & 0.952 & 0.994\tabularnewline
\hline 
10 & 1.468 & 0.725 & 0.805 & 1.068 & 2.151 & 6.433 & 1.136 & 0.651 & 0.973 & 0.999\tabularnewline
\hline 
\multicolumn{11}{|c|}{Simulation of Pareto distribution ($\alpha=4$, samples $10^{5}$)}\tabularnewline
\hline 
samp. size & $(n-1)/n$ & \multicolumn{5}{c|}{eB empirical quantiles ($n$ times)} & sample & \multicolumn{3}{c|}{probability exceeding tail}\tabularnewline
\hline 
$n$ & $F$-quantile & min & 0.01 & median & 0.99 & max & mean & $a=1$ & $a=3$ & $a=5$\tabularnewline
\hline 
1000 & 5.624 & 0.514 & 0.684 & 1.097 & 3.159 & 16.001 & 1.227 & 0.633 & 0.949 & 0.994\tabularnewline
\hline 
100 & 3.163 & 0.526 & 0.688 & 1.098 & 3.174 & 14.698 & 1.227 & 0.635 & 0.954 & 0.995\tabularnewline
\hline 
10 & 1.779 & 0.622 & 0.722 & 1.106 & 3.154 & 27.162 & 1.237 & 0.652 & 0.972 & 1.000\tabularnewline
\hline 
\multicolumn{11}{|c|}{Simulation of Pareto distribution ($\alpha=2$, samples $10^{5}$)}\tabularnewline
\hline 
samp. size & $(n-1)/n$ & \multicolumn{5}{c|}{eB empirical quantiles ($n$ times)} & sample & \multicolumn{3}{c|}{probability exceeding tail}\tabularnewline
\hline 
$n$ & $F$-quantile & min & 0.01 & median & 0.99 & max & mean & $a=1$ & $a=3$ & $a=5$\tabularnewline
\hline 
1000 & 31.623 & 0.28 & 0.467 & 1.201 & 9.920 & 939.63 & 1.793 & 0.633 & 0.951 & 0.994\tabularnewline
\hline 
100 & 10.000 & 0.269 & 0.472 & 1.201 & 9.971 & 441.51 & 1.761 & 0.634 & 0.952 & 0.994\tabularnewline
\hline 
10 & 3.163 & 0.37 & 0.521 & 1.221 & 9.872 & 398.84 & 1.793 & 0.651 & 0.972 & 0.999\tabularnewline
\hline 
\end{tabular}
}
\end{table}

Table \ref{tab:T1} shows the behavior of the empirical upper bound
(\ref{eq:bound}), $eB(\nu)$, in samples of size $10$, $100$ and
$1000$ in the half-normal, exponential and Pareto I distributions
for the values of the parameter $\alpha=2$ (infinite variance), $\alpha=4$
(infinite kurtosis) and $\alpha=6$. We used standard versions of the
distributions because the inequalities are not affected by changes
of scale.

For each combination of distribution function $F$ and sample size
$n$ we calculated quantile $q_{1}$ so that $1-F(q_{1})=1/n$. At
this point $eB(q_{1})$ must be an upper bound of the tail function
 in approximately $63\%$ of the cases, see ($3$) in Proposition \ref{prop:P10}
and Corollary \ref{cor:Cor}. Equivalently, $n\,eB(q_{1})=X_{n,n}/q_{1}$
must be an upper bound of $1$.

For each distribution we simulated $10^{5}$ samples of size $n$ and
calculated the maximum $X_{n,n}$, thus obtaining a sample of size
$10^{5}$ for $n\,eB(q_{1})$. For this sample we calculated the mean
and the quantiles $0$, $0.1$, $0.5$, $0.99$ and $1$. We see in
the median and mean that $eB(q_{1})$ slightly overestimates the tail
function. With a sample size of $100$, the median for the Pareto
distribution ($\alpha=6$) is $1.064$, overestimating $1$ by $6.4\%$.
In all the tables these are average values and the result is very
similar in all sample sizes and similar in all distributions. Only
Pareto ($\alpha=2$) increases to 1.20.

The probabilities, obtained through simulation, of exceeding the tail
function for $a$ equal to $1$, $3$ and $5$, coincide with those
of Corollary \ref{cor:Cor}. Table \ref{tab:T1} shows that, if
the inequality (\ref{eq:iBa}) is safer than (\ref{eq:iB1}) for $a=3$
and $a=5$, the median and mean of the samples of
$a\,eB(q_{1})$ are 3 or 5 times larger. For the Pareto distribution
($\alpha=6$), with sample size $n=100$, they are $3.192$ and $5.320$,
thus usually overestimating $1$ by extremely large quantities.
This prompts us to propose $a=1$.

Looking at the $0.01$ quantile of the $eB$ samples makes the above
proposal even more plausible. It shows that while (\ref{eq:bound})
underestimates in $37\%$ of cases, only $1\%$ of the cases (excluding
Pareto with $\alpha=2$ and some samples with $n=10$) are below $0.70$
of the true value. In other words, if by chance underestimates, it does so with quite realistic values.
\begin{table}

\caption{\protect\label{tab:T2}Probability of exceeding the tail value simulating
Pareto I distributions with $10^{5}$ replicates for $\mu=1$, different
parameter $\alpha$ and sample sizes $n$.}
\resizebox{\columnwidth}{!}{%
\begin{tabular}{|c|c|c|c|c|c|c|c|c|c|c|c|}
\hline 
\multicolumn{4}{|c|}{$\alpha=4$, $n=1000$} & \multicolumn{4}{c|}{$\alpha=3$, $n=1000$} & \multicolumn{4}{c|}{$\alpha=2$, $n=1000$}\tabularnewline
\hline 
\hline 
Probability & 0.999 & 0.9995 & 0.9998 & Probability & 0.999 & 0.9995 & 0.9998 & Probability & 0.999 & 0.9995 & 0.9998\tabularnewline
\hline 
Quantile & 5.624 & 6.688 & 8.409 & Quantile & 10.000 & 12.600 & 17.100 & Quantile & 31.623 & 44.722 & 70.711\tabularnewline
\hline 
Tail (n times) & 1.000 & 0.500 & 0.200 & Tail (n times) & 1.000 & 0.500 & 0.200 & Tail (n times) & 1.000 & 0.500 & 0.200\tabularnewline
\hline 
Median & 1.096 & 0.922 & 0.733 & Median & 1.131 & 0.898 & 0.661 & Median & 1.197 & 0.846 & 0.535\tabularnewline
\hline 
Pro. exceeds & 0.633 & 1.000 & 1.000 & Pro. exceeds & 0.632 & 0.983 & 1.000 & Pro. exceeds & 0.633 & 0.866 & 0.994\tabularnewline
\hline 
\multicolumn{4}{|c|}{$\alpha=4$, $n=100$} & \multicolumn{4}{c|}{$\alpha=3$, $n=100$} & \multicolumn{4}{c|}{$\alpha=2$, $n=100$}\tabularnewline
\hline 
Probability & 0.99 & 0.995 & 0.998 & Probability & 0.99 & 0.995 & 0.998 & Probability & 0.99 & 0.995 & 0.998\tabularnewline
\hline 
Quantile & 3.163 & 3.761 & 4.729 & Quantile & 4.642 & 5.849 & 7.938 & Quantile & 10 & 14.143 & 22.361\tabularnewline
\hline 
Tail (n times) & 1.000 & 0.500 & 0.200 & Tail (n times) & 1.000 & 0.500 & 0.200 & Tail (n times) & 1.000 & 0.500 & 0.200\tabularnewline
\hline 
Median & 1.097 & 0.923 & 0.734 & Median & 1.131 & 0.897 & 0.661 & Median & 1.203 & 0.851 & 0.538\tabularnewline
\hline 
Pro. exceeds & 0.633 & 1.000 & 1.000 & Pro. exceeds & 0.634 & 0.983 & 1.000 & Pro. exceeds & 0.635 & 0.868 & 0.995\tabularnewline
\hline 
\multicolumn{4}{|c|}{$\alpha=4$, $n=10$} & \multicolumn{4}{c|}{$\alpha=3$, $n=10$} & \multicolumn{4}{c|}{$\alpha=2$, $n=10$}\tabularnewline
\hline 
Probability & 0.9 & 0.95 & 0.98 & Probability & 0.9 & 0.95 & 0.98 & Probability & 0.9 & 0.95 & 0.98\tabularnewline
\hline 
Quantile & 1.779 & 2.115 & 2.66 & Quantile & 2.155 & 2.715 & 3.685 & Quantile & 3.163 & 4.473 & 7.072\tabularnewline
\hline 
Tail (n times) & 1.000 & 0.500 & 0.200 & Tail (n times) & 1.000 & 0.500 & 0.200 & Tail (n times) & 1.000 & 0.500 & 0.200\tabularnewline
\hline 
Median & 1.105 & 0.929 & 0.739 & Median & 1.142 & 0.907 & 0.668 & Median & 1.223 & 0.865 & 0.547\tabularnewline
\hline 
Pro. exceeds & 0.652 & 1.000 & 1.000 & exceeds & 0.652 & 0.995 & 1.000 & Pro. exceeds & 0.651 & 0.893 & 1.000\tabularnewline
\hline 
\end{tabular}
}
\end{table}

Table \ref{tab:T2} shows the behavior of the empirical upper bound,
$eB(\nu)$, (\ref{eq:bound}) beyond the $(n-1)/n$ quartile ($\nu>q_{1}$)
in samples of size $10$, $100$ and $1000$ for the Pareto I distribution,
while simulating $10^{5}$ replicates with $\mu=1$ and parameter $\alpha=2$
(infinite variance), $\alpha=3$ (infinite skewness) and $\alpha=4$
(infinite kurtosis). Since the tail function of these distributions
decreases as $\nu^{-\alpha}$ with $\alpha>1$ and the empirical upper
bound decreases as $\nu^{-1}$, the probabilities that $eB(\nu)$
exceeds $1-F\left(\nu\right)$ in the inequality (\ref{eq:iB1})
should increase with $\nu>q_{1}$.

We have selected distributions with a gradual tail decline in order
to assess this phenomenon, and quantiles $q_{1}$, $q_{0.5}$
and $q_{0.2}$ corresponding to probabilities $(n-1)/n$ ,$(n-0.5)/n$
and $(n-0.2)/n$ , which  for $n=100$ result in $0.99$, $0.995$ and
$0.998$. At these quantiles, $n$ times the tail function, $n(1-F)$,
is $1$, $0.5$ and $0.2$, as shown in Table \ref{tab:T2}.

The median of $eB(\nu)$ samples and the probability of exceeding
the tail value are determined in each of the three quantiles following
the simulation of $10^{5}$ samples of each distribution and sample
size. It may be observed that the results change very little between
the sample sizes, but vary with respect to the parameter $\alpha$. 

The probability that $eB(\nu)$ exceeds the tail value in $\alpha=4$,
which is approximately $63\%$ at $q_{1}$, increases to $100\%$
at $q_{0.5}$ and $q_{0.2}$. The probability that it exceeds in $\alpha=3$,
is approximately $98\%$ at $q_{0.5}$ and $100\%$ at $q_{0.2}$.
The probability that it exceeds in $\alpha=2$, is approximately $87\%$
at $q_{0.5}$ and $99\%$ at $q_{0.2}$. Additionally, Table \ref{tab:T2}
shows that the higher the quantile, the safer the use of inequality
(\ref{eq:iB1}), as the Median always decreases more slowly than the
value of the tail, provided that the decline is slower than the Tail.

This fact is closely related to the number of finite moments, as shown
in (\ref{eq:momen}), which determines the order of decrease of the
tail function. The drop is quick enough in distributions with four
or more finite moments, meaning that the inequality may be applied $100\%$
of the times across quantiles $q_{0.5}$ and $q_{0.2}$. These quantiles are connected to the maxima of samples of sizes $2n$ and $5n$, as $q_{1}$ is connected to the maxima of samples of size $n$ in Proposition \ref{prop:P10}.

\section{\protect\label{sec:DJI}A real-world data set}

\begin{figure}
\caption{\protect\label{fig:F1_DJI}Dow Jones Industrial Average daily prices
between $1$ January $2015$ and $31$ December $2022$}
\includegraphics[scale=0.6]{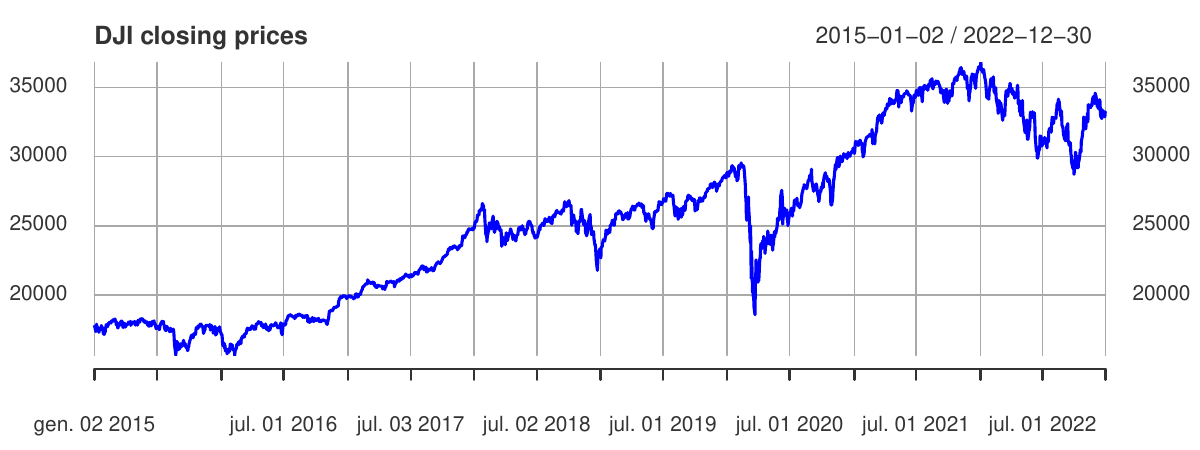}
\end{figure}
The field of finance, insurance, and risk theory has witnessed significant
advancements in the extreme value theory, see \cite{McNeil} and \cite{Longin}.
Even so, alternative models are still being sought to avoid estimates that are 
unacceptable in practice. \cite{Mandelbrot} used the Dow Jones Industrial
Average (DJI), which is one of the first equity indices used to represent
the entire US stock market, to demonstrate the limitations of the Gaussian
random variables, particularly with regard to risk measurement, and he proposed
Pareto distribution-based models, such as those in Section \ref{subsec:Pareto}.
Avoiding model risk is critical when there are multiple competing
models in a real-world situation. In these cases inequality (\ref{eq:iB1})
is relevant, as it serves to set the maximum probability of the possible
misfortunes.

Just a few days after Coronavirus was declared a pandemic by the World
Health Organization, major stock markets lost more than $15\%$ of
their market capitalization. This study begins with an easy query
about the stock market crisis of $2020$: What is the probability
of an investor's portfolio losing more than $15\%$ or $20\%$ of
its value in a single day?

We addressed this issue by analyzing the behavior of the DJI daily prices
between $1$ January $2015$ and $31$ December $2022$ (eight years),
a period which included the $2020$ stock market crisis, which we downloaded with the
R package quantmod, \cite{Jeffrey}. Figure \ref{fig:F1_DJI} shows
the daily closing prices, denoted by $S_{t}$. Throughout the paper,
we used the daily returns of the previous data set, described as

\[
R_{t}=100\,log\left(S_{t}/S_{t-1}\right).
\]
The fundamental model for financial assets postulates that returns
are distributed as independent Gaussian random variables and, equivalently,
the logarithm of asset prices follows a Brownian motion. However,
this assumption is far from perfect from an empirical perspective.
The estimate $\hat{\mu}=0.031$, $\hat{\sigma}=1.187$ is obtained
by fitting the $n=2013$ DJI returns into a Gaussian distribution, $N\left(\mu,\sigma^{2}\right)$.
The probability of losing more than $10\%$ of the DJI price using
the Gaussian model for our data set is $Pr\left\{ R<-10\right\} =1.46\cdot10^{-17}$.
This is once every $2.74\cdot10^{14}$ years, assuming that there are
$250$ business days in a year. This figure is significantly greater
than the Earth's estimated age of $4,543$ billion years. However, during this period, there were
losses of more than $10\%$ on $12$ and $16$ March
$2020$ ($10.52$ and $13.84$ respectively),
see Table \ref{tab:DJIloss}. In $8$ years, $2$
incidents were recorded, which refutes the model's applicability.

The DJI dataset contains $2013$ daily returns, of which $1081$ are
larger than zero, $930$ less than zero and $2$ equal to zero. The
set of $930$ returns less than zero with the sign changed will be
denoted by negative returns. We will start by applying the inequality
(\ref{eq:iB1}) to the DJI negative returns ($N$), which are a sample
of a positive random variable. The conditional probability over a
given threshold $\nu>x_{n,n}=13.84$ is $Pr\left\{ N>\nu\left|N\right.\right\} \leq13.84/(930 \nu)$,
hence, the unconditional probability is
\begin{equation}
Pr\left\{ -R>\nu\right\} =Pr\left\{ N>\nu\left|N\right.\right\} Pr\left\{ N\right\} \leq\frac{13.84}{930 \nu}\frac{930}{2013}=\frac{13.84}{2013 \nu},\label{eq:unconditional}
\end{equation}
where the probability of a negative return is estimated by $Pr\left\{ N\right\} =Pr\left\{ R<0\right\} =930/2013$.
Therefore, the probabilities over thresholds $13.84$, $15$ and $20$,
are $4.97\cdot10^{-4}$, $4.58\cdot10^{-4}$ and $3.44\cdot10^{-4}$,
respectively. The return periods measured in years, which corresponds
to the inverse of the probability divided by $250$, over thresholds
$13.84$, $15$ and $20$ are $8.05$, $8.73$ and $11.63$, see Table
\ref{tab:NP_tau}. 

Figure \ref{fig:LP_fit} shows the empirical tail function (the step
function that jumps down by $1/n$ at each of the $n$ data points)
for daily DJI losses over threshold 4.5 and the empirical upper bound
$eB\left(\nu\right)$ from the estimation of the $q_{1}$ quantile,
$x_{n-1,n}=10.52$, to $20$. This curve is connected to the sample
maximum $x_{n,n}$ and surpasses the three approximations of the LPD
models which we will discuss below.

\begin{table}
\caption{\protect\label{tab:DJIloss}Dow Jones Industrial losses on a daily,
monthly, and annual basis (2015-2022).}
\resizebox{\columnwidth}{!}{%
\begin{tabular}{|c|c|c|c|c|c|c|c|c|}
\hline 
\multicolumn{9}{|c|}{Largest percentage losses of DJI on a daily, monthly, and annual basis}\tabularnewline
\hline 
\hline 
\multicolumn{3}{|c|}{Largest daily losses} & \multicolumn{3}{c|}{Largest monthly losses} & \multicolumn{3}{c|}{Largest yearly losses}\tabularnewline
\hline 
Rank & Date & Change & Rank & Date & Change & Rank & Date & Change\tabularnewline
\hline 
1 & 2020-03-16 & 13.842 & 1 & 2020-03-31 & 13.842 & 1 & 2020-12-31 & 13.842\tabularnewline
\hline 
2 & 2020-03-12 & 10.523 & 2 & 2020-06-30 & 7.148 & 2 & 2018-12-31 & 4.714\tabularnewline
\hline 
3 & 2020-03-09 & 8.106 & 3 & 2018-02-28 & 4.714 & 3 & 2022-12-30 & 4.021\tabularnewline
\hline 
4 & 2020-06-11 & 7.148 & 4 & 2020-04-30 & 4.544 & 4 & 2015-12-31 & 3.640\tabularnewline
\hline 
5 & 2020-03-18 & 6.510 & 5 & 2020-02-28 & 4.518 & 5 & 2016-12-30 & 3.447\tabularnewline
\hline 
6 & 2020-03-11 & 6.034 & 6 & 2022-09-30 & 4.021 & 6 & 2019-12-31 & 3.093\tabularnewline
\hline 
7 & 2018-02-05 & 4.714 & 7 & 2015-08-31 & 3.640 & 7 & 2021-12-31 & 2.560\tabularnewline
\hline 
8 & 2020-03-20 & 4.653 & 8 & 2022-05-31 & 3.631 & 8 & 2017-12-29 & 1.793\tabularnewline
\hline 
\end{tabular}
}
\end{table}

By comparing the maximums of the $96$ months considered, Table \ref{tab:DJIloss} demonstrates that the maximum DJI
monthly loss, 13.842, is a very significant maximum that clearly exceeds
the median of maximums. Comparing the 8 annual maximums reinforces
this idea. This might be sufficient to validate the inequality-based
analysis (\ref{eq:iB1}).

Three models of location Pareto distribution (\ref{eq:LPD}) have
been adjusted to the DJI negative returns, using extreme value theory.
In one case, following the technique proposed by Clauset, the minimum
threshold selected was $\mu=1.75$. The residual variation coefficient
technique was applied in the remaining two cases, see
\cite{Castillo19}. By default, the \emph{thrselect} function (R
package {\emph{ercv}}) selects the lowest threshold (closest to
$0$) with a $p$-value greater than $0.10$, in this case the {minimum
of the} sample. Here we also propose a variation of the algorithm's
last phase, which involves selecting the threshold with the best
fit in terms of the biggest $p$-value, which in this case $\mu=1.4$.

Based on the three selected thresholds we estimated the other two
parameters $\left(\alpha,\beta\right)$ for each LPD model considered.
These values are found in the last three columns of Table \ref{tab:LPD}.
The estimated $\alpha$ parameter suggests four finite moments for
the LPD-1 model, three finite moments (infinite kurtosis) for the
LPD-2 model, and two finite moments (infinite skewness) for the LPD-3
model. 

The computed parameters for negative returns, $N$, using the peak-over-threshold
method, give us the conditional probabilities of returns less than
zero. Then, with a threshold of $\nu>0$, the unconditional probability
of daily returns are calculated using the equation (\ref{eq:unconditional}).
\begin{figure}
\caption{\protect\label{fig:LP_fit}Empirical tail function and three approximations
by LPD model, with the parameters given in Table \ref{tab:LPD}, and
empirical upper bound eB.}
\includegraphics[scale=0.65]{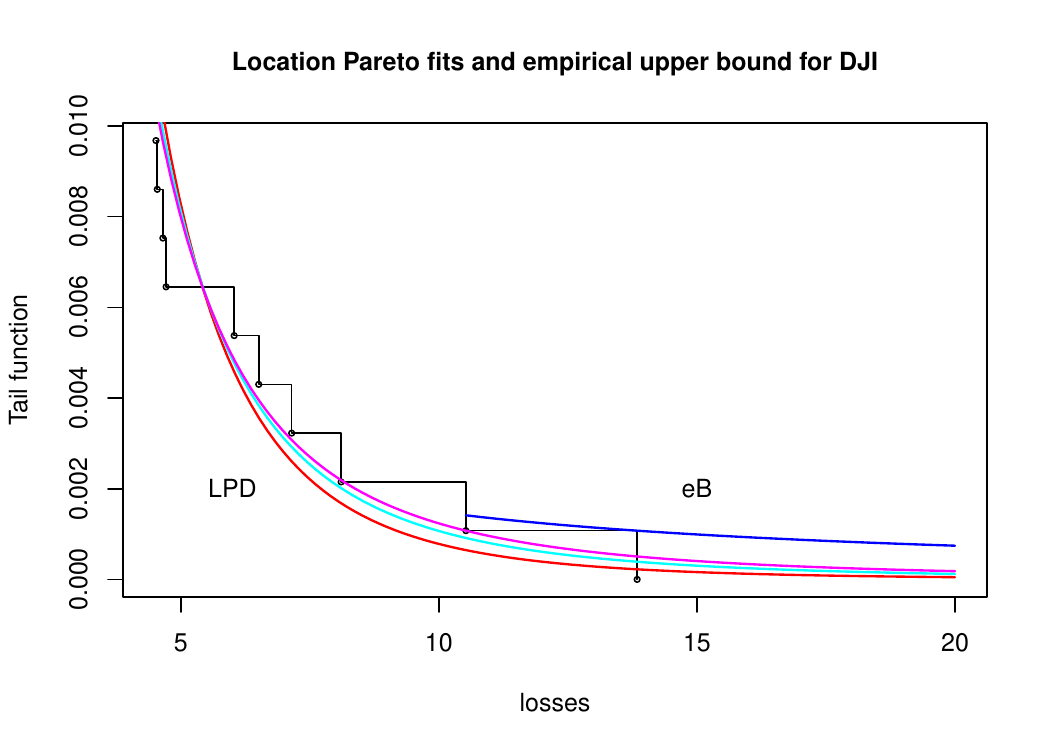}
\end{figure}

\begin{table}
\caption{\protect\label{tab:NP_tau}The estimated probabilities for Dow Jones
Industrial losses by empirical upper bound method.}
\begin{tabular}{|c|c|c|c|}
\hline 
size = 2013 & \multicolumn{3}{c|}{Empirical upper bound}\tabularnewline
\hline 
Threshold & 13.84 & 15 & 20\tabularnewline
\hline 
Probability & 4.97·$10^{-04}$ & 4.58·$10^{-04}$ & 3.44·$10^{-04}$\tabularnewline
\hline 
Return period & 8.05 & 8.73 & 11.63\tabularnewline
\hline 
\end{tabular}
\end{table}

Table \ref{tab:LPD} shows the estimated probabilities for losses
above thresholds $5$, $10$, $13.84$, $15$ and $20$ for the three
LPD models, according to previously selected threshold parameters
$0.000$, $1.400$ and $1.750$, corresponding to subsamples of size
$930$, $154$ and $104$ respectively. To enable comparisons with
empirically observed losses, Table \ref{tab:LPD} contains the
expected value of the number of losses above the thresholds over an
$8$-year period, as well as the return period in years. The comparisons
show that the higher the threshold (the smaller the subsample size)
the closer the estimates are to the observed values.
\begin{table}

\caption{\protect\label{tab:LPD}The estimated probabilities for Dow Jones
Industrial losses by location Pareto distribution
(LPD).}
\begin{tabular}{|c|c|c|c|c|c|c|c|c|}
\hline 
\multicolumn{9}{|c|}{Location Pareto distributions}\tabularnewline
\hline 
\hline 
\multicolumn{6}{|c|}{Loss probabilities greater than the threshold} & \multicolumn{3}{c|}{Parameters}\tabularnewline
\hline 
Threshold & 5 & 10 & 13.84 & 15 & 20 & alpha & beta & mu\tabularnewline
\hline 
LPD-1 & 0.0038 & 0.0004 & 0.0001 & 7.48E-05 & 2.28E-05 & 4.825 & 0.610 & 0.000\tabularnewline
\hline 
LPD-2 & 0.0037 & 0.0005 & 0.0002 & 1.40E-04 & 5.60E-05 & 3.389 & 0.740 & 1.400\tabularnewline
\hline 
LPD-3 & 0.0037 & 0.0006 & 0.0002 & 1.87E-04 & 8.44E-05 & 2.831 & 0.745 & 1.750\tabularnewline
\hline 
Empirical & 0.0030 & 0.0010 & 0.0005 & 0 & 0 &  &  & \tabularnewline
\hline 
\multicolumn{6}{|c|}{Expected value  of losses for a new period of 8 years} & \multicolumn{3}{c|}{Parameters}\tabularnewline
\hline 
Threshold & 5 & 10 & 13.84 & 15 & 20 & alpha & beta & mu\tabularnewline
\hline 
LPD-1 & 7.69 & 0.73 & 0.21 & 0.15 & 0.05 & 4.825 & 0.610 & 0.000\tabularnewline
\hline 
LPD-2 & 7.54 & 0.99 & 0.36 & 0.28 & 0.11 & 3.389 & 0.740 & 1.400\tabularnewline
\hline 
LPD-3 & 7.42 & 1.15 & 0.47 & 0.38 & 0.17 & 2.831 & 0.745 & 1.750\tabularnewline
\hline 
Empirical & 6 & 2 & 1 & 0 & 0 &  &  & \tabularnewline
\hline 
\multicolumn{6}{|c|}{Return period in years} & \multicolumn{3}{c|}{Parameters}\tabularnewline
\hline 
Threshold & 5 & 10 & 13.84 & 15 & 20 & alpha & beta & mu\tabularnewline
\hline 
LPD-1 & 1.05 & 11.06 & 38.75 & 53.50 & 175.21 & 4.825 & 0.610 & 0.000\tabularnewline
\hline 
LPD-2 & 1.07 & 8.10 & 22.17 & 28.56 & 71.38 & 3.389 & 0.740 & 1.400\tabularnewline
\hline 
LPD-3 & 1.08 & 7.00 & 17.10 & 21.36 & 47.42 & 2.831 & 0.745 & 1.750\tabularnewline
\hline 
Empirical & 1.34 & 4.03 & 8.05 &  &  &  &  & \tabularnewline
\hline 
\end{tabular}

\end{table}

Figure \ref{fig:LP_fit} shows the three approximations of the LPD
models with the parameters given in Table \ref{tab:LPD} and the empirical
upper bound $eB\left(\nu\right)$ from (\ref{eq:iB1}), for daily
DJI losses over threshold 4.5. The empirical tail function is a stair
steps function which fall at the daily DJI losses (see Table \ref{tab:DJIloss}).
There are only ten jumps that exceed this cutoff, six of which are
above 5.

The LPD-3 model exhibits a very good fit to the tail function above
threshold 5. Given that the estimate is based on $104$ observations,
a sufficiently large sample size, this model is the one we suggest
for extrapolation. LPD-3 is the closest to the empirical upper bound
$eB\left(\nu\right)$ of the three models.

Adopting a model with little data, on the other hand, is easier than
accepting one with a large amount of data. As a result, the more data with
which we can accept a model, the more evidence we have that it is
valid. In this sense, the LPD-1 model estimated with $930$ observations
is  the most plausible overall (we can reasonably assume four finite
moments to model the data). It should be observed that the fit to the data reverses
at threshold 5 as opposed to the more tail-focused LPD-3 model.

In the estimated Pareto models, LPD-2 retains an intermediate position,
and the outcomes of the three models are quite comparable. This verifies
the process, which holds up well with a sizable data set and in which 
the threshold selection has little impact.

Of the three estimates, the empirical upper limit $eB$ produces
findings that are most similar to the LPD-3 model; this could result
in this model being chosen. The return period of this model is extrapolated
to be $47.42$ years at the $\nu=20$ threshold. Given that the DJI
experienced a loss of $22.6$ on Black Monday in October $1987$,
which occurred $37$ years ago, this seems quite likely.

\section*{\protect\label{sec:Disc}Discussion }

Classical non-parametric methods, which are based on the distribution
of order statistics and yield model-free results, do not provide information
beyond the maximum observed in a sample. Markov's inequality shows
how more information may be obtained by considering distributions with a specific number of finite moments. An improved version of Markov's
inequality emphasizes that it must be applied from the maximum of
a sample using the partial expectation, \cite{Cohen}.

In actuality, partial expectation cannot be calculated; however,
it may be estimated by $X_{n,n}/n$, where $X_{n,n}$ is the maximum
of an $n$-sample. This is the origin of the empirical upper bound
(\ref{eq:bound}) and inequalities (\ref{eq:iB1}) and (\ref{eq:iBa}).
Our goal has been to examine the distribution and application of this
statistic, by simply assuming a finite expectation, without introducing
any model uncertainty.

The inequality (\ref{eq:iB1}) appears to be the best upper bound
in the worst case scenario based on the numerical data presented in
Section \ref{sec:Num}. At the maximum, in mean and median, the empirical
upper bound $eB(X_{n,n})$ only slightly overestimates the tail function.
And where it underestimates, it does so with quite realistic values. Table
\ref{tab:T2} shows that the higher the quantile, the safer the use
of inequality.

Additionally, given a dataset, it is possible to confirm that the inequality (\ref{eq:iB1})
overestimates by a complementary examination. For instance, in Section
\ref{sec:DJI}, the estimated values of quantile $q_{1}$ with the
three LPD models estimated, whose parameters appear in Table \ref{tab:LPD},
are $9.361$, $10.154$ and $10.805$. All three values are below
the maximum of DJI daily losses $13.842$. This fact, according to
Proposition \ref{prop:P10}, guarantees that (\ref{eq:iB1}) exceeds
the tail function.

In situations when model risk is crucial, this innovative method offers
safety, clarity, and simplicity. It acts as a guide for the worst-case
situation that might occur.

\section*{acknowledgments}
This work was funded by the Grant/Award Number: PID2022-137414OB-I00; Severo Ochoa and María de Maeztu Program for Centers and Units of 
Excellence in R\&D, Grant/Award Number: CEX2020-001084-M.

\bibliographystyle{plainnat}
\bibliography{Bib_ProNoRisk}

\begin{thebibliography}{19}
\providecommand{\natexlab}[1]{#1}
\providecommand{\url}[1]{\texttt{#1}}
\expandafter\ifx\csname urlstyle\endcsname\relax
  \providecommand{\doi}[1]{doi: #1}\else
  \providecommand{\doi}{doi: \begingroup \urlstyle{rm}\Url}\fi

\bibitem[Arnold(2015)]{Arnold}
Barry Arnold.
\newblock \emph{Pareto distribution}.
\newblock Taylor \& Francis, Boca Raton, FL, 2015.

\bibitem[Beirlant et~al.(2006)Beirlant, Goegebeur, Segers, and
  Teugels]{Beirlant}
Jan Beirlant, Yuri Goegebeur, Johan Segers, and Jozef~L Teugels.
\newblock \emph{Statistics of extremes: theory and applications}.
\newblock John Wiley \& Sons, 2006.

\bibitem[Castillo(1988)]{ecastillo88}
Enrique Castillo.
\newblock \emph{Extreme value theory in engineering}.
\newblock Academic Press, San Diego, 1988.

\bibitem[Clauset et~al.(2009)Clauset, Shalizi, and Newman]{Clauset}
Aaron Clauset, Cosma Shalizi, and Mark Newman.
\newblock Power-law distributions in empirical data.
\newblock \emph{SIAM review}, 51\penalty0 (4):\penalty0 661--703, 2009.

\bibitem[Cohen(2015)]{Cohen}
Joel Cohen.
\newblock Markov's inequality and chebyshev's inequality for tail
  probabilities: A sharper image.
\newblock \emph{The American Statistician}, 69\penalty0 (1):\penalty0 5--7,
  2015.

\bibitem[Coles(2001)]{Coles}
Stuart Coles.
\newblock \emph{An Introduction to Statistical of Extremes Values}.
\newblock Springer, 2001.

\bibitem[David and Nagaraja(2004)]{David}
Herbert David and Haikady Nagaraja.
\newblock \emph{Order statistics}.
\newblock John Wiley \& Sons, 2004.

\bibitem[Davison and Huser(2015)]{davison15}
Anthony Davison and Rapha{\"e}l Huser.
\newblock Statistics of extremes.
\newblock \emph{Annual Review of Statistics and its Application}, 2\penalty0
  (1):\penalty0 203--235, 2015.

\bibitem[de~Carvalho(2016)]{Carvalho}
Miguel de~Carvalho.
\newblock Statistics of extremes: Challenges and opportunities.
\newblock \emph{Extreme events in finance: A handbook of extreme value theory
  and its applications}, pages 195--213, 2016.

\bibitem[del Castillo and Puig(2023)]{Castillo23}
Joan del Castillo and Pedro Puig.
\newblock Distinguishing between a power law and a pareto distribution.
\newblock \emph{Physical Review E}, 107\penalty0 (6):\penalty0 064113, 2023.

\bibitem[del Castillo et~al.(2019)del Castillo, Serra, Padilla, and
  Mori{\~n}a]{Castillo19}
Joan del Castillo, Isabel Serra, Maria Padilla, and David Mori{\~n}a.
\newblock Fitting tails by the empirical residual coefficient of variation: The
  ercv package.
\newblock \emph{R J.}, 11\penalty0 (2):\penalty0 56, 2019.

\bibitem[Gibbons and Chakraborti(2014)]{Gibbons}
Jean Gibbons and Subhabrata Chakraborti.
\newblock \emph{Nonparametric statistical inference: revised and expanded}.
\newblock CRC press, 2014.

\bibitem[Gomes and Guillou(2015)]{Gomes}
M~Ivette Gomes and Armelle Guillou.
\newblock Extreme value theory and statistics of univariate extremes: a review.
\newblock \emph{International statistical review}, 83\penalty0 (2):\penalty0
  263--292, 2015.

\bibitem[Longin(2016)]{Longin}
Fran{\c{c}}ois~(ed.) Longin.
\newblock \emph{Extreme events in finance: A handbook of extreme value theory
  and its applications}.
\newblock John Wiley \& Sons, 2016.

\bibitem[Mandelbrot(2001)]{Mandelbrot}
Benoit Mandelbrot.
\newblock Scaling in financial prices: I. tails and dependence.
\newblock \emph{Quantitative Finance}, 1\penalty0 (1):\penalty0 113, 2001.

\bibitem[McNeil et~al.(2005)McNeil, Frey, and Embrechts]{McNeil}
A.~McNeil, R.~Frey, and P.~Embrechts.
\newblock \emph{Quantitative Risk Management: Concepts, Techniques and Tools}.
\newblock Princeton Series in Finance, New Jersey, 2005.

\bibitem[Ogasawara(2021)]{Ogasawara}
Haruhiko Ogasawara.
\newblock Improvements of the markov and chebyshev inequalities using the
  partial expectation.
\newblock \emph{Communications in Statistics-Theory and Methods}, 50\penalty0
  (1):\penalty0 116--131, 2021.

\bibitem[Ryan and Ulrich(2022)]{Jeffrey}
Jeffrey Ryan and Joshua Ulrich.
\newblock \emph{quantmod: Quantitative Financial Modelling Framework}, 2022.
\newblock URL \url{https://CRAN.R-project.org/package=quantmod}.
\newblock R package version 0.4.20.

\bibitem[Wilks(1941)]{Wilks}
Samuel Wilks.
\newblock Determination of sample sizes for setting tolerance limits.
\newblock \emph{The Annals of Mathematical Statistics}, 12\penalty0
  (1):\penalty0 91--96, 1941.

\end{thebibliography}

\end{document}